\documentclass[english,aps,pra,amsmath,amssymb,showpacs,twocolumn]{revtex4-1}
\usepackage[T1]{fontenc}
\usepackage[latin9]{inputenc}
\usepackage{amsthm}
\usepackage{amsmath}
\usepackage{amssymb}
\usepackage{graphicx}

\makeatletter
 
 \@ifundefined{textcolor}{}
 {%
   \definecolor{BLACK}{gray}{0}
   \definecolor{WHITE}{gray}{1}
   \definecolor{RED}{rgb}{1,0,0}
   \definecolor{GREEN}{rgb}{0,1,0}
   \definecolor{BLUE}{rgb}{0,0,1}
   \definecolor{CYAN}{cmyk}{1,0,0,0}
   \definecolor{MAGENTA}{cmyk}{0,1,0,0}
   \definecolor{YELLOW}{cmyk}{0,0,1,0}
 }
\theoremstyle{plain}
\newtheorem{thm}{\protect\theoremname}

\makeatother

\usepackage{babel}
\providecommand{\theoremname}{Theorem}

\begin{document}

\title{Quantum Parrondo's games constructed by quantum random walks}

\author{Min Li}

\email{mickeylm@mail.ustc.edu.cn}

\author{Yong-Sheng Zhang}

\email{yshzhang@ustc.edu.cn}

\author{Guang-Can Guo}

\address{Key Laboratory of Quantum Information, University of Science and
Technology of China, CAS, Hefei, 230026, People's Republic of China}
\begin{abstract}
We construct a Parrondo's game using discrete time quantum walks.
Two lossing games are represented by two different coin operators.
By mixing the two coin operators $U_{A}(\alpha_{A},\beta_{A},\gamma_{A})$
and $U_{B}(\alpha_{B},\beta_{B},\gamma_{B})$, we may win the game.
Here we mix the two games in position instead of time. With a number
of selections of the parameters, we can win the game with sequences
ABB, ABBB, \emph{et al}. If we set $\beta_{A}=45^{\circ},\gamma_{A}=0,\alpha_{B}=0,\beta_{B}=88^{\circ}$,
we find the game 1\emph{ }with {\normalsize $U_{A}^{S}=U^{S}(-51^{\circ},45^{\circ},0)$,
$U_{B}^{S}=U^{S}(0,88^{\circ},-16^{\circ})$ will win and get the
most profit. }If we set $\alpha_{A}=0,\beta_{A}=45^{\circ},\alpha_{B}=0,\beta_{B}=88^{\circ}$
and{\normalsize{} the game 2 with $U_{A}^{S}=U^{S}(0,45^{\circ},-51^{\circ})$,
$U_{B}^{S}=U^{S}(0,88^{\circ},-67^{\circ})$, will win most. And }game
1\emph{ }{\normalsize is equivalent to the }game\emph{ }2\emph{ }with
the changes of sequences and steps. But at a large enough steps, the
game will loss at last.
\end{abstract}
\maketitle

\section{Introduction}

Parrondo's games present an apparently paradoxical situation where
individually losing games can be combined to win \cite{Harmer,Parrondo-2000,Harmer-2001}.
Parrondo's games have important applications in many physical and
biological systems. For example in control theory, the analogy of
Parrondo's games can be used to design a second order switched mode
circuit which is unstable in either mode but is stable when switched
\cite{Allison-2001}. Recentely, quantum version of Parrondo's games
was also introduced in \cite{Flitney-2002,Chandrashekar-2011,Ampadu-2011,Flitney-2003,Flitney-2012}. 

Quantum walks(QWs), as the quantum version of the classical random
walks, were first introduced in 1993 \cite{Aharonov} (For overviews,
see \cite{Kempe,Venegas-2012}). According to the time evolution,
QWs can be devided into discrete-time (DTQWs) and continuous-time
\cite{Farhi-1998} QWs (CTQWs). A number of quantum algorithms based
on QWs have already been proposed in \cite{Childs2003,Shenvi2003,Childs2002,Childs2004,Ambainis,Ambainis2004},
and QWs are found to be universal for quantum computation \cite{Childs2009,Lovett-2010}.

In this paper, we construct a Parrondo's games using QWs. Two models
fo Parrondo's game using QWs were introduced in \cite{Chandrashekar-2011,Flitney-2012},
but both of them are switched according to the time. Here we construct
the games with coin operators alternated depending on the position
in one step.

\section{Discrete time quantum walks on a line}

In this paper, we concern with the discrete-time quantum walks. To
be consistent, we adopt analogous definitions and notations as those
outlined in \cite{Travaglione}. The total Hilbert space is given
by $\mathcal{H}\equiv\mathcal{H}_{P}\otimes\mathcal{H}_{C}$, where
$\mathcal{H}_{P}$ is spanned by the orthonormal vectors $\left\{ \mid x\rangle:x\in\mathbb{Z}\right\} $
which represents the position of the walker and $\mathcal{H}_{C}$
is the Hilbert space of chirality, or ``coin'' states, spanned by
the orthonormal basis $\{\mid L\rangle,\mid R\rangle\}$.

Each step of the quantum walk can be split into two operations: the
flip of a coin and the position motion of the walker according to
the coin state.

Here, for simplicity, we choose a Hadamard coin as the normal quantum
walk's coin, so the coin operator can be written as 

\begin{equation}
\hat{H}\mid\downarrow\rangle=\frac{1}{\sqrt{2}}(\mid\downarrow\rangle+\mid\uparrow\rangle),\:\hat{H}=\frac{1}{\sqrt{2}}\left(\begin{array}{cc}
1 & 1\\
1 & -1
\end{array}\right).
\end{equation}

The position displacement operator is given by

\begin{equation}
\hat{S}=e^{i\hat{p}\hat{\sigma}_{z}}=\sum_{x}\hat{S}_{x},
\end{equation}
where $\hat{p}$ is the momentum operator, $\hat{\sigma}_{z}$ is
the Pauli-$z$ operator, 

\begin{equation}
\hat{S}_{x}=\mid x+1\rangle\langle x\mid\otimes\mid\uparrow_{x}\rangle\langle\uparrow_{x}\mid+\mid x-1\rangle\langle x\mid\otimes\mid\downarrow_{x}\rangle\langle\downarrow_{x}\mid.
\end{equation}

Therefore, the state of the walker after $N$ steps is given by

\begin{equation}
\begin{aligned}\mid\Psi_{N}\rangle & =\left[\hat{S}(\hat{I}_{P}\otimes\hat{H}_{C})\right]^{N}\mid\Psi_{0}\rangle\\
 & =\left[\sum_{x}\hat{S}_{x}(\hat{I}_{P}\otimes\hat{H}_{C})\right]^{N}\mid\Psi_{0}\rangle,
\end{aligned}
\label{eq:hadamardDQW-1}
\end{equation}
where $\mid\Psi_{0}\rangle$ is the initial state of the system. 

For generalized DTQWs, we use a U(2) matrix:
\begin{equation}
U_{\alpha,\beta,\gamma,\theta}=e^{i\theta}\left(\begin{array}{cc}
e^{i\alpha}\cos\beta, & -e^{-i\gamma}\sin\beta\\
e^{i\gamma}\sin\beta, & e^{-i\alpha}\cos\beta
\end{array}\right)\label{eq:U(2)}
\end{equation}
instead of the Hadamard operator in Eq. (\ref{eq:hadamardDQW-1}).
We can easily know the QWs using a SU(2) operator:
\begin{equation}
U_{\alpha,\beta,\gamma}^{S}=\left(\begin{array}{cc}
e^{i\alpha}\cos\beta, & -e^{-i\gamma}\sin\beta\\
e^{i\gamma}\sin\beta, & e^{-i\alpha}\cos\beta
\end{array}\right),\label{eq:SU(2)}
\end{equation}
 have the same properties as using a U(2) coin operator\cite{Li-average}.
In this paper we always use the SU(2) operator.

\section{Parrondo's game using QRW with position dependent on coin operator}

Parrondo's games arise in the following situation when we have two
games that are losing when played seperately, but the two games played
in combination will form an overall winning game. 

Here we present a scheme of a player playing games using DTQW. He
has two games A and B, but he do not play the two games alternated
according to the time, but the position \cite{Li-potential}. Alternating
the coin operation step-by-step has been discussed in \cite{Flitney-2012}.
Here we discuss Parrondo's games using QRWs with the coin operation
alternated with the position in one step.

\begin{figure}
\begin{centering}
\includegraphics[width=3in]{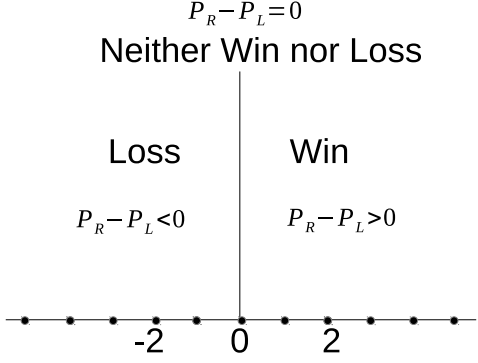}
\par\end{centering}

\caption{Pictorial illustration of the conditions for win or loss for QWs on
a line.}
\label{fig:win-loss}
\end{figure}

The game is constructed as follows:
\begin{itemize}
\item Both game A and B are represented by different quantum operators $U(\alpha_{A},\beta_{A},\gamma_{A})$
and $U(\alpha_{B},\beta_{B},\gamma_{B})$.
\item The state is in $\mid\Psi_{0}\rangle=\frac{1}{\sqrt{2}}\mid0\rangle(\mid\downarrow\rangle+i\mid\uparrow\rangle)$
initially.
\item Game \emph{A} and \emph{B} are played alternately in different positions
in one step, instead of step by step. i.e. game A is played on site
$x=nq$ and game B is played on site $x\neq nq$. The evolution operator
can be written as:
\begin{equation}
U=\sum_{x=nq,n\in Z}\hat{S}_{x}U(\alpha_{A},\beta_{A},\gamma_{A})+\sum_{x\neq nq,n\in Z}\hat{S}_{x}U(\alpha_{B},\beta_{B},\gamma_{B}),
\end{equation}
where $q$ is the period, $n$ is an integer, and the final state
after $N$ steps is given by
\begin{equation}
\mid\Psi_{N}\rangle=U^{N}\mid\Psi_{0}\rangle.
\end{equation}
For $q=3$, it means we play games with the sequence ABBABB on the
line.
\end{itemize}
As denoted in Fig. \ref{fig:win-loss}, after $N$ steps, if the probability
$P_{R}$ of the walker to be found in the right of the origin, is
greater than the probability $P_{L}$ in the left of the origin, that
is $P_{R}-P_{L}>0$, we consider the player win $P_{R}-P_{L}$. Similarly,
if $P_{R}-P_{L}<0$, the player losses $P_{L}-P_{R}$. If $P_{R}-P_{L}=0$,
it means the player neither losses nor wins.

\section{Results}

Here we use $P_{R}-P_{L}=\sum_{x>0}(P(x)-P(-x))$, where $P(x)$ is
the probability of the particle to be found at $x$, but not the average
positon $\left\langle x\right\rangle $ to indicate the player win
or loss, because sometimes $\left\langle x\right\rangle =\sum_{x>0}x(P(x)-P(-x)$)
may be positive, but $P_{R}-P_{L}$ be negtive, and vice versa. So
we can not use $\left\langle x\right\rangle $ to indicate the player
win or loss.
\begin{thm}
\label{thm:pr-pl}If the initail state $\mid\Psi_{0}\rangle=1/\sqrt{2}(\mid0L\rangle+i\mid0R\rangle)$,
after $t$ steps quantum walk: $P_{R}-P_{L}=M(\beta,t)\sin(\alpha+\gamma)$,
where $M(\beta,t)$ only depends on $\beta$ and $t$.\end{thm}
\begin{proof}
The same as the proof of $\left\langle x\right\rangle =G(\beta,t)\sin(\alpha+\gamma)$
from \cite{Li-average}, we can easily get $P_{R}-P_{L}=M(\beta,t)\sin(\alpha+\gamma)$.
\end{proof}
Theorem \ref{thm:pr-pl} shows $P_{R}-P_{L}=M(\beta,t)\sin(\alpha+\gamma)$,
then if we set $\alpha+\gamma=\pi/2$, we can calculate $M(\beta,t)$
varied with the change of $\beta$. Fig. \ref{fig:0-beta-90} shows
the $P_{R}-P_{L}$ varied with the change of $\beta$, after 100 steps
QWs with the initial state $\mid\Psi_{0}\rangle=1/\sqrt{2}(\mid0L\rangle+i\mid0R\rangle)$,
and the coin operator $U^{S}=U^{S}(0,\beta,90)$ (Here 90 denote 90
degree, in the following, we always use this setting). The figure
inside shows $P_{R}-P_{L}$ varies with the parameter $\beta\in[80,100]$.
From the figure, we can know that when $\beta\approx88,$ $P_{R}-P_{L}$
gets its maximun value. In the following , we set $\beta_{A}=45,\beta_{B}=88$.

\begin{figure}
\begin{centering}
\includegraphics[width=3in]{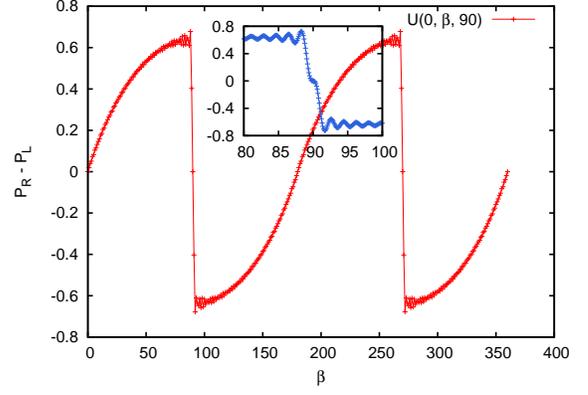}
\par\end{centering}

\caption{(Color online) $P_{R}-P_{L}$ of the walker after $t=100$ steps QWs
for the initial state $\mid\Psi_{0}\rangle=1/\sqrt{2}(\mid0L\rangle+i\mid0R\rangle)$,
and the coin operator $U^{S}=U^{S}(0,\beta,90)$. The figure inside
shows $P_{R}-P_{L}$ vary with parameter $\beta\in[80,100]$.}

\label{fig:0-beta-90}
\end{figure}

Fig. \ref{fig:alpha} shows the $P_{R}-P_{L}$ of the walker after
$100$ steps QWs with $q=3$ (a sequence of games \emph{ABB}), and
$U_{A}^{S}=U^{S}(15,45,30)$, $U_{B}^{S}=U^{S}(\alpha_{B},88,0)$.
From the figure, we can know that $P_{R}-P_{L}$ does not vary withh
the change of $\alpha_{B}$. So in the following, we always calculate
with $\alpha_{B}=0$.

\begin{figure}
\begin{centering}
\includegraphics[width=3in]{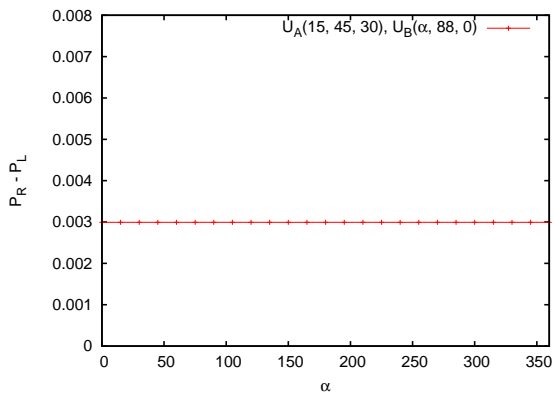}
\par\end{centering}

\caption{(Color online) $P_{R}-P_{L}$ of the walker after $100$ steps QWs
for the initial state $\mid\Psi_{0}\rangle=1/\sqrt{2}(\mid0L\rangle+i\mid0R\rangle)$,
with $q=3$ (for a sequence of games \emph{ABB}), and $U_{A}^{S}=U^{S}(15,45,30)$,
$U_{B}^{S}=U^{S}(\alpha_{B},88,0)$.}

\label{fig:alpha}
\end{figure}

\begin{figure}
\begin{centering}
\includegraphics[width=3in]{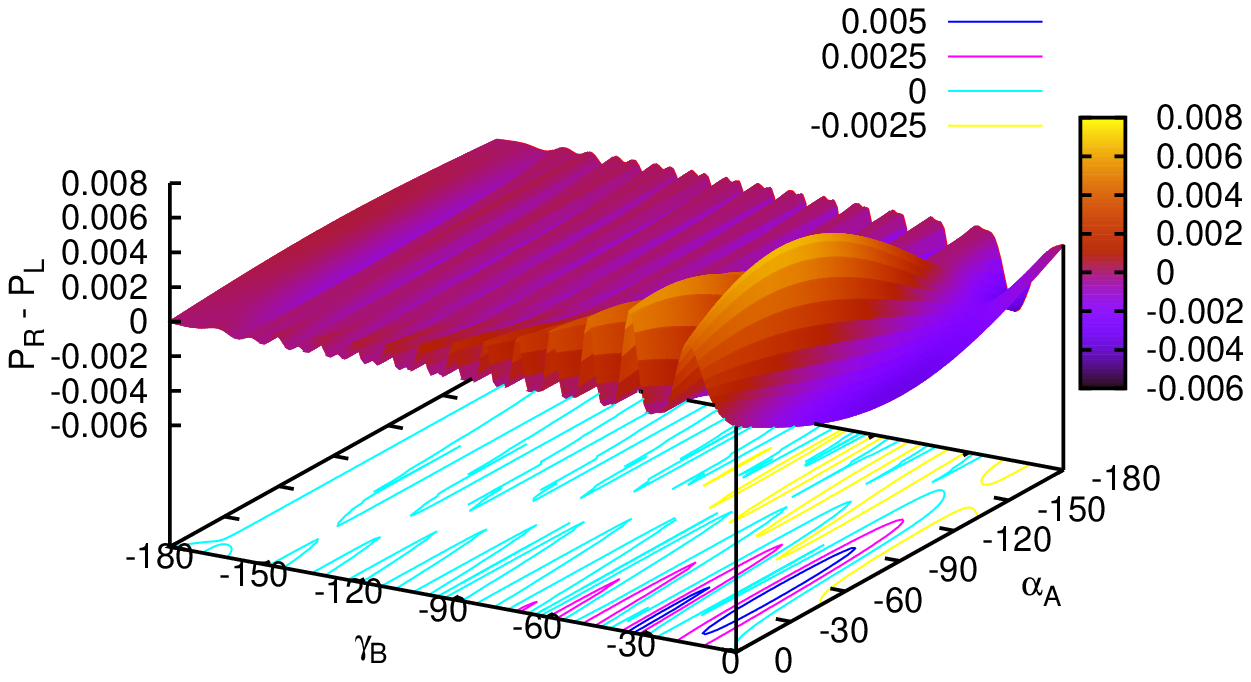}
\par\end{centering}

\begin{centering}
\includegraphics[width=3in]{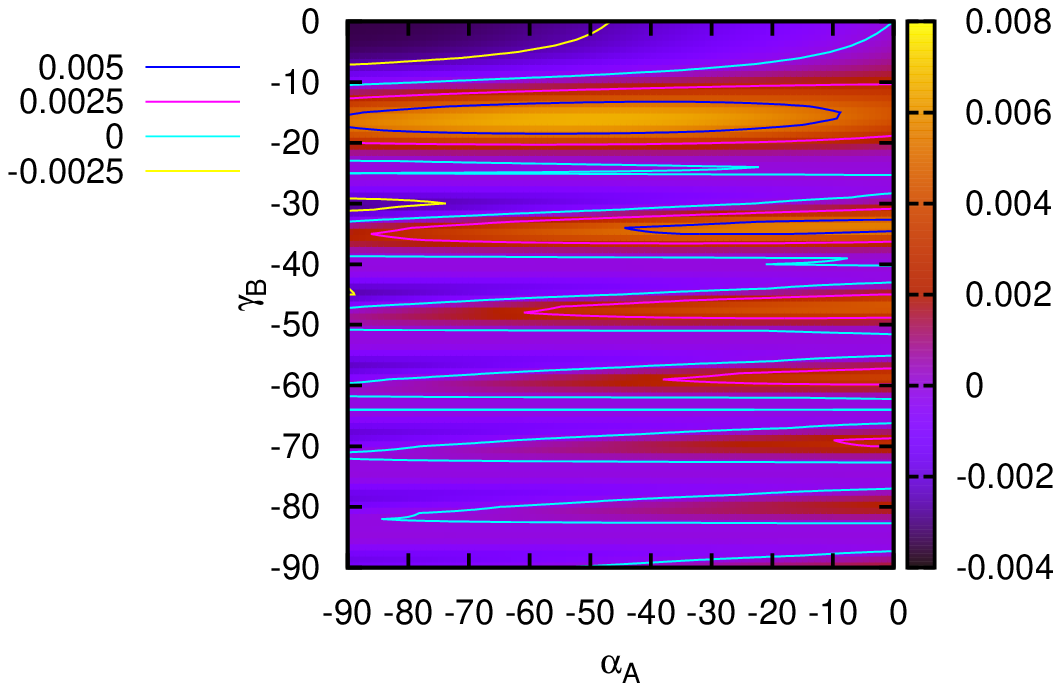}
\par\end{centering}

\caption{(Color online) $P_{R}-P_{L}$ of the walker after $100$ steps QWs
for the initial state $\mid\Psi_{0}\rangle=1/\sqrt{2}(\mid0L\rangle+i\mid0R\rangle)$,
with $q=3$ (for a sequence \emph{ABB}), and $U_{A}^{S}=U^{S}(\alpha_{A},45,0)$,
$U_{B}^{S}=U^{S}(0,88,\gamma_{B})$, where $\alpha_{A},\gamma_{B}\in[-180,0].$
The bottom of the left figure is the contour line of $P_{R}-P_{L}$.
The right figure shows the contour line of $P_{R}-P_{L}$, when $\alpha_{A},\gamma_{B}\in[-90,0].$}

\label{fig-45-88-1}
\end{figure}

\begin{figure}
\begin{centering}
\includegraphics[width=3in]{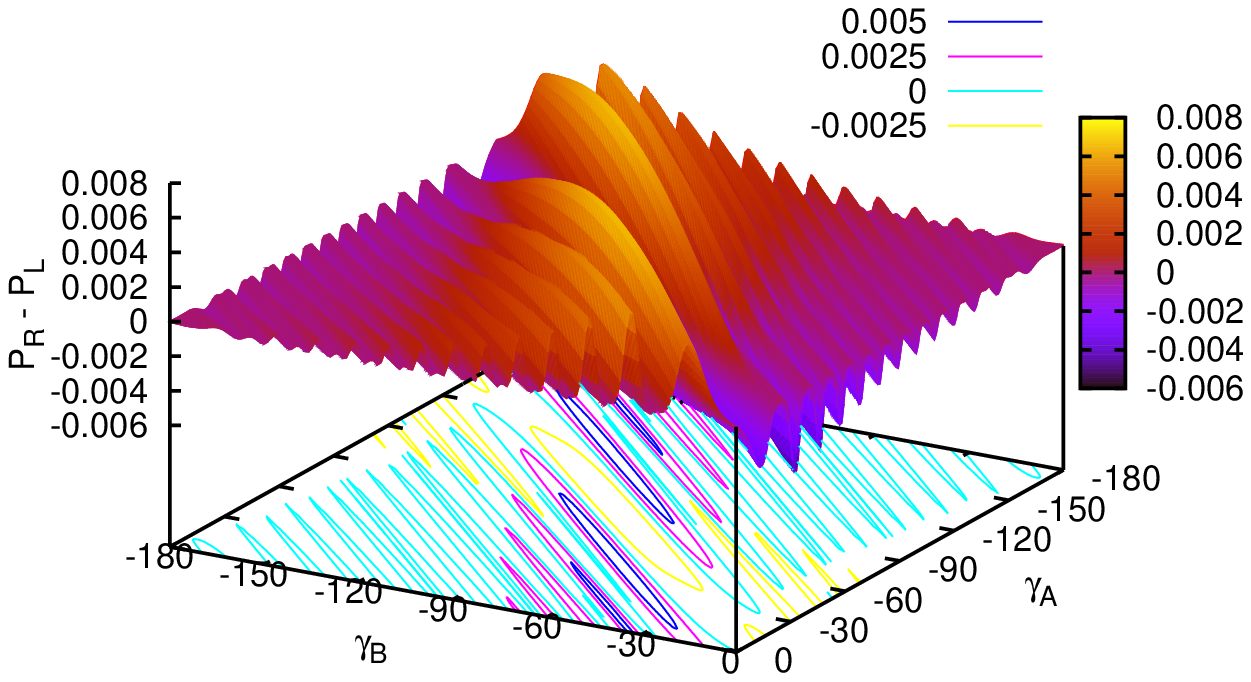}
\par\end{centering}

\begin{centering}
\includegraphics[width=3in]{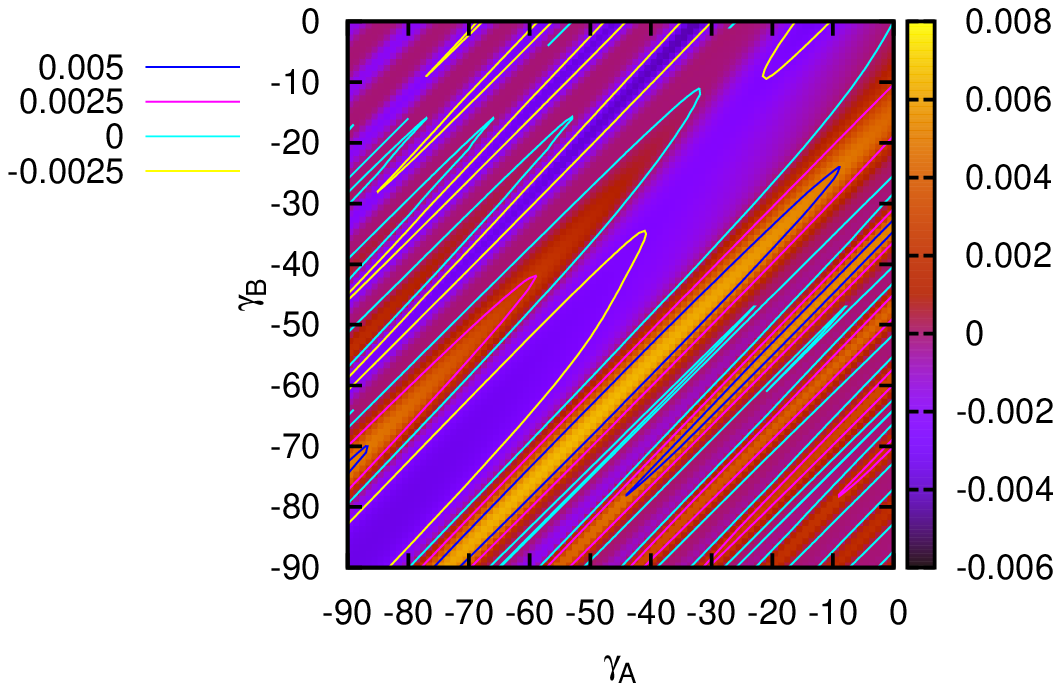}
\par\end{centering}

\caption{(Color online) $P_{R}-P_{L}$ of the walker after $100$ steps QWs
for the initial state $\mid\Psi_{0}\rangle=1/\sqrt{2}(\mid0L\rangle+i\mid0R\rangle)$,
with $q=3$ (for a sequence \emph{ABB}), and $U_{A}^{S}=U^{S}(0,45,\gamma_{A})$,
$U_{B}^{S}=U^{S}(0,88,\gamma_{B})$, where $\gamma_{A},\gamma_{B}\in[-180,0].$
The bottom of the left figure is the contour line of $P_{R}-P_{L}$.
The right figure shows the contour line of $P_{R}-P_{L}$, when $\gamma_{A},\gamma_{B}\in[-90,0].$}

\label{fig-45-88-2}
\end{figure}

Fig. \ref{fig-45-88-1} shows $P_{R}-P_{L}$ of the walker after $100$
steps QWs for the initial state $\mid\Psi_{0}\rangle=1/\sqrt{2}(\mid0L\rangle+i\mid0R\rangle)$,
with $q=3$, and $U_{A}^{S}=U^{S}(\alpha_{A},45,0)$, $U_{B}^{S}=U^{S}(0,88,\gamma_{B})$,
where $\alpha_{A},\gamma_{B}\in[-180,0].$ From Theorem \ref{thm:pr-pl},
we can know that after \emph{$100$} steps QWs with any $\alpha_{A},\gamma_{B}\in[-180,0]$,
game \emph{A} and game \emph{B} are lossing when played seperately,
but Fig. \ref{fig-45-88-1} shows that there are many choice of $\alpha_{A}$
and $\gamma_{B}$ that can win the game (i.e. $P_{R}-P_{L}>0$) with
sequence of games \emph{ABB, }i.e. the Parrondo's paradox arises in
quantum version. The maximun of $P_{R}-P_{L}\approx0.00673$, when
$U_{A}^{S}=U^{S}(-51,45,0)$, $U_{B}^{S}=U^{S}(0,88,-16)$.

The same as above, Fig. \ref{fig-45-88-2} shows $P_{R}-P_{L}$ of
the walker after $100$ steps QWs with a sequence of games \emph{ABB}
and $U_{A}^{S}=U^{S}(0,45,\gamma_{A})$, $U_{B}^{S}=U^{S}(0,88,\gamma_{B})$,
where $\gamma_{A},\gamma_{B}\in[-180,0].$ $\gamma_{A},\gamma_{B}\in[-180,0]$.
Game \emph{A} and game \emph{B} are lossing when played seperately,
but combining them as ABB, we will win with many choices of $\gamma_{A}$
and $\gamma_{B}$. The maximun of $P_{R}-P_{L}\approx0.00673$, when
$U_{A}^{S}=U^{S}(0,45,-51)$, $U_{B}^{S}=U^{S}(0,88,-67)$.

\begin{figure}
\begin{centering}
\includegraphics[width=3in]{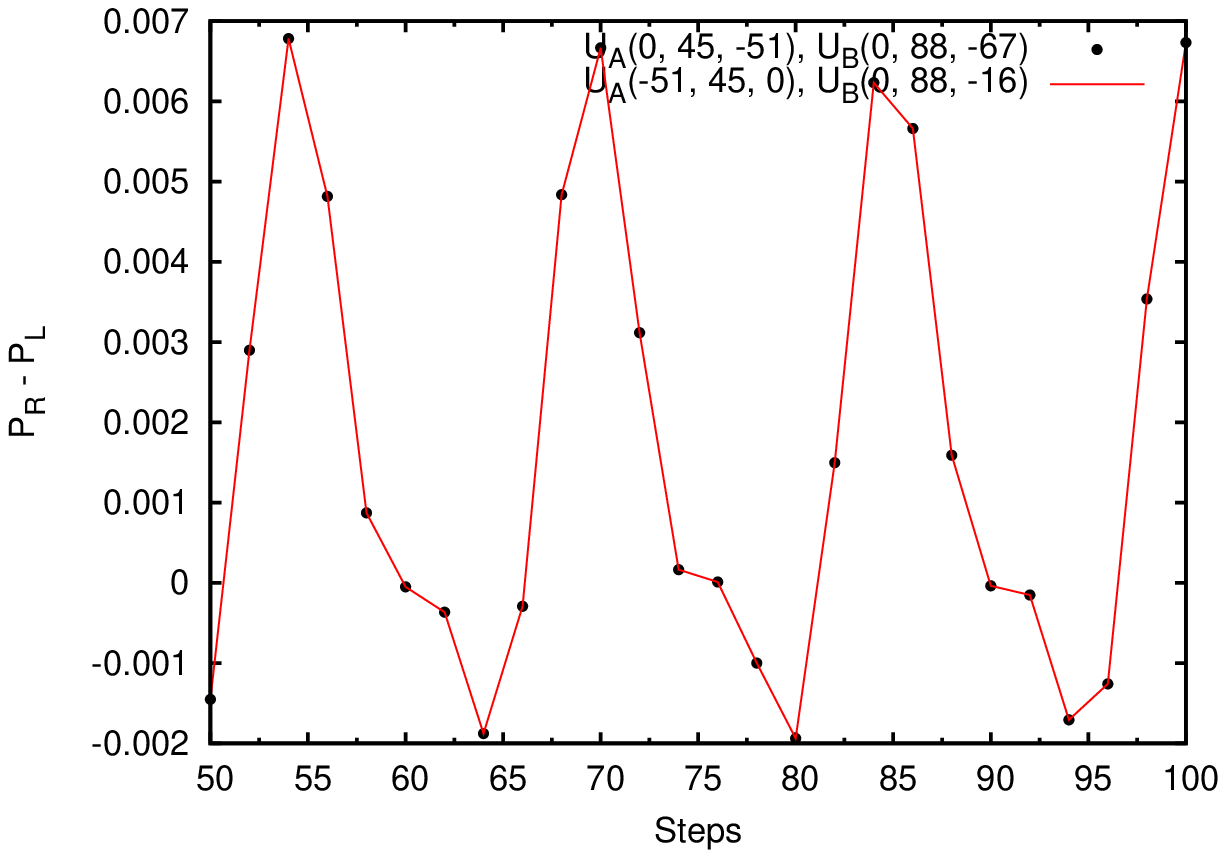}
\par\end{centering}

\begin{centering}
\includegraphics[width=3in]{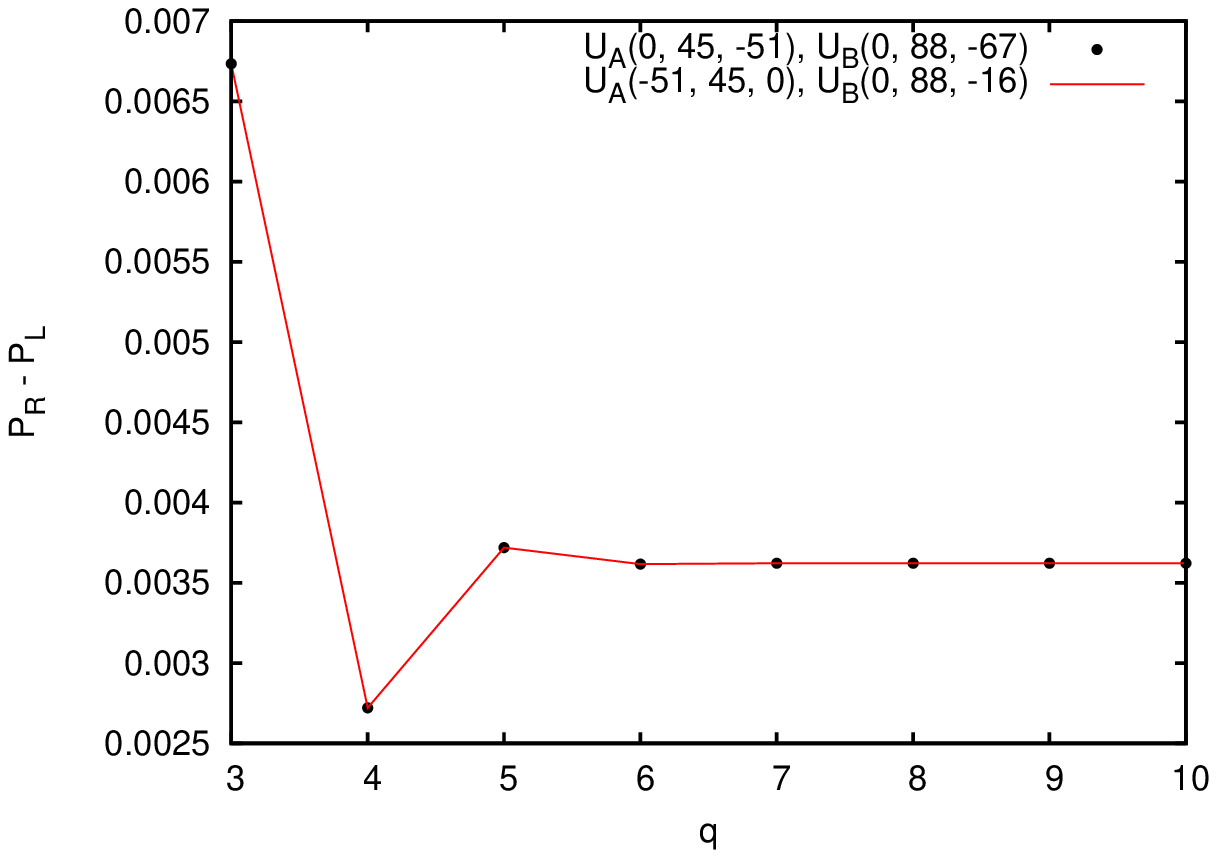}\caption{(Color online) Left figure: $P_{R}-P_{L}$ of the walker after QWs
different steps with $q=3$. Right figure: $P_{R}-P_{L}$ of the walker
after $100$ steps QWs with different $q$ (sequences of games, e.
x. $q=4,$\emph{ ABBB}). $100$ steps with (red line) $U_{A}^{S}=U^{S}(-51,45,0)$
and $U_{B}^{S}=U^{S}(0,88,-16)$ or (black point) $U_{A}^{S}=U^{S}(0,45,-51)$
and $U_{B}^{S}=U^{S}(0,88,-67)$.}

\par\end{centering}

\label{fig:step-period-1}
\end{figure}

Fig. \ref{fig:step-period-1} shows the above two games of the maximun
$P_{R}-P_{L}$, game 1\emph{:} with \emph{$U_{A1}^{S}=U^{S}(-51,45,0)$,
$U_{B1}^{S}=U^{S}(0,88,-16)$; }games 2\emph{ }with\emph{ $U_{A2}^{S}=U^{S}(0,45,-51)$,
$U_{B2}^{S}=U^{S}(0,88,-67)$.} From the left figure, we know that
game 1 and game 2 have the\emph{ s}ame result of win or loss with
different steps. The right figure shows that\emph{ }game 1\emph{ }and\emph{
}game\emph{ }2\emph{ }will win in the same way with the increasing
of $q$. So\emph{ }game 1\emph{ }is equivalent to game 2, then in
the following, we only need to study the game 1.

\begin{figure}
\begin{centering}
\includegraphics[width=3in]{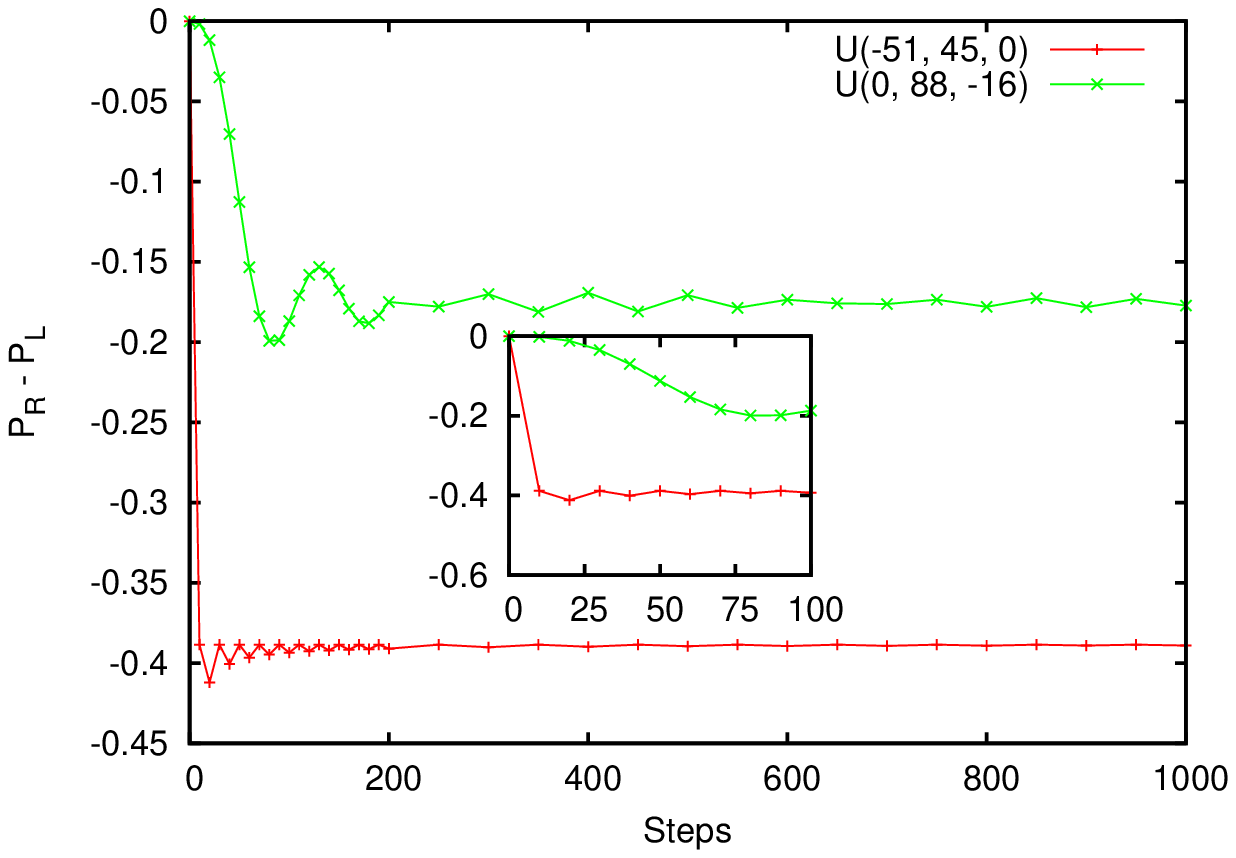}
\par\end{centering}

\caption{(Color online) $P_{R}-P_{L}$ of the walker for QWs after $t$ steps,
with initial state $\mid\Psi_{0}\rangle=1/\sqrt{2}(\mid0L\rangle+i\mid0R\rangle)$,
and coin operator $U^{S}(-51,45,0)$ (red line) or $U^{S}(0,88,-16)$
(green line). }

\label{fig: 51-45 and 88-16}
\end{figure}

\begin{figure}
\begin{centering}
\includegraphics[width=3in]{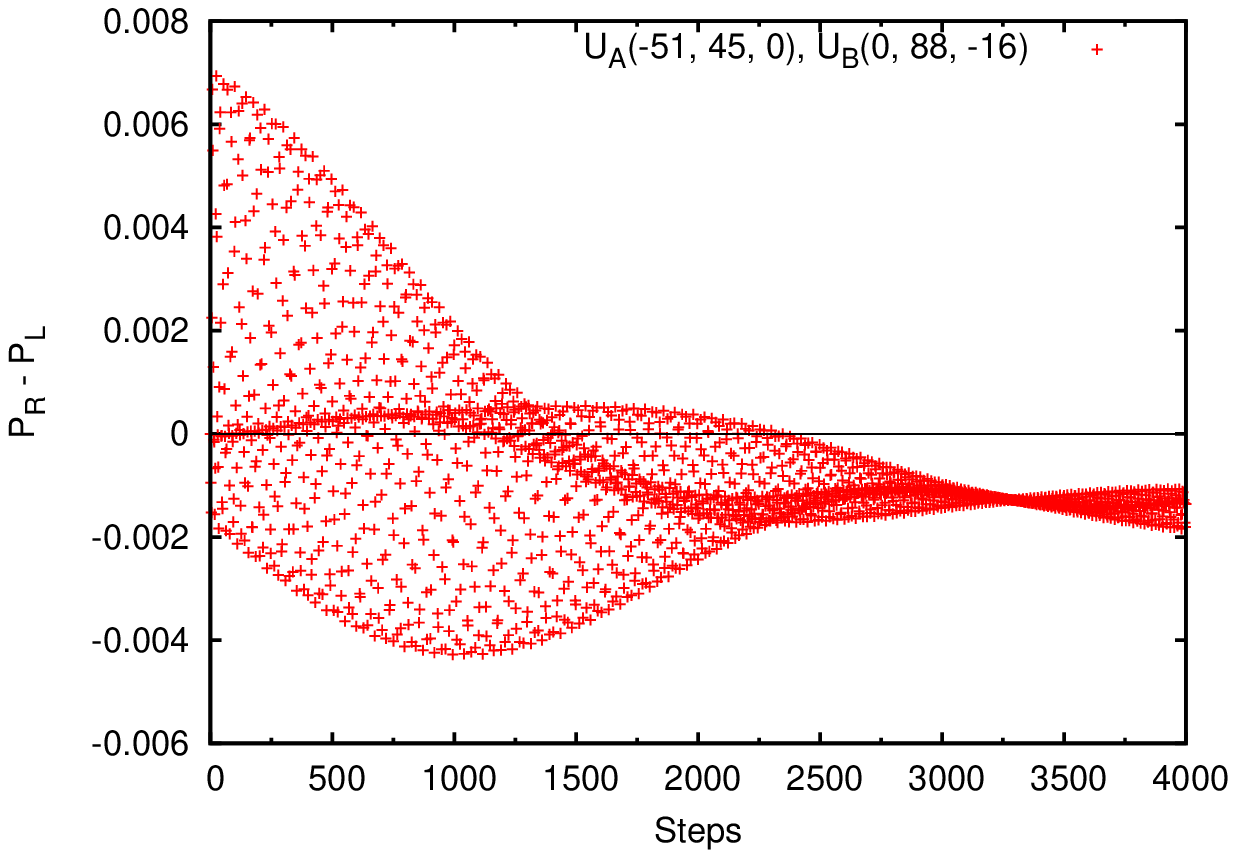}
\par\end{centering}

\caption{(Color online) $P_{R}-P_{L}$ of the walker after QWs different steps
with $q=3$, $U_{A}^{S}=U^{S}(-51,45,0)$, $U_{B}^{S}=U^{S}(0,88,-16)$.
(only even steps)}

\label{fig-45-88-3}
\end{figure}

Next, we want to know for the sequence of game \emph{ABB} whether
there still exists the effect of Parrondo's paradox with the increacing
of steps. First, we need to know games A\emph{ }with  $U_{A}^{S}=U^{S}(-51,45,0)$
and\emph{ }B with $U_{B}^{S}=U^{S}(0,88,-16)$ still loss after different
steps QWs? Fig. \ref{fig: 51-45 and 88-16} shows that the games\emph{
A} and \emph{B} still loss after different steps QWs, and the $P_{R}-P_{L}$
decreases fast and will tend to stable with the increasing of steps.
Second, Fig. \ref{fig-45-88-3} shows the combined game in the situation
of sequence \emph{ABB }situation, with the increasing of steps (only
even steps), the result of the game will always fluctuate, and will
loss with a large enough step.

\section{Conclusion}

In this paper, we have constructed Parrondo's games by using the one-dimensional
discrete time quantum walks. The game is constructed by two lossing
games \emph{A} and \emph{B} with two different biased coin operators
$U_{A}(\alpha_{A},\beta_{A},\gamma_{A})$ and $U_{B}(\alpha_{B},\beta_{B},\gamma_{B})$,
but different from the time dependent sequences of games in \cite{Flitney-2012},
here we consider the position dependent sequences of games. With a
number of selections of $\alpha_{A},\beta_{A},\gamma_{A},\alpha_{B},\beta_{B},\gamma_{B}$,
we can form a winning game with sequences \emph{ABB, ABBB}, \emph{et
al}. If we set $\beta_{A}=45,\gamma_{A}=0,\alpha_{B}=0,\beta_{B}=88$,
we find the \emph{game 1 }with $U_{A}^{S}=U^{S}(-51,45,0)$, $U_{B}^{S}=U^{S}(0,88,-16)$
will win most. If we set $\alpha_{A}=0,\beta_{A}=45,\alpha_{B}=0,\beta_{B}=88$,
the \emph{game 2} with $U_{A}^{S}=U^{S}(0,45,-51)$, $U_{B}^{S}=U^{S}(0,88,-67)$
will win most. And \emph{game 1 }is equivalent to the \emph{game}
\emph{2 }with the changes of sequences and steps. But at a large enough
steps, the game will loss.
\begin{acknowledgments}
This work was supported by the National Natural Science Foundation
of China (Grant No. 10974192, 61275122), the National Fundamental
Research Program of China (Grant No. 2011CB921200, 2011CBA00200),
K. C. Wong Education Foundation and CAS.\end{acknowledgments}

\end{document}